\documentclass[5p]{elsarticle}
\usepackage[caption=false,font=footnotesize]{subfig}
\usepackage{fixltx2e}
\usepackage{stfloats}

\usepackage[table]{xcolor}
\usepackage{lipsum}
\usepackage{amsmath}
\usepackage{amsfonts}
\usepackage{amssymb}
\usepackage{enumitem}
\usepackage{listings}
\usepackage[hyphens]{url}
\usepackage{graphicx}
\usepackage{stfloats}
\usepackage{xcolor}
\usepackage{amsthm}
\usepackage{color} 

\usepackage[caption=false,font=footnotesize]{subfig}

\definecolor{light-gray}{gray}{0.8}
\lstdefinestyle{c_bash}{
 language=bash,
 commentstyle=\color{black},
 columns=flexible
}

\lstdefinestyle{c_text}{
  backgroundcolor=\color{white},
  breaklines=true,
  captionpos=b,
  columns=flexible,
  extendedchars=true,
  frame=single,
  keepspaces=true,
  language={},
  numbers=left,
  xleftmargin=10pt,
  numbersep=5pt,
  numberstyle=\small\color{black},
  rulecolor=\color{black},
  rulesepcolor=\color{black},
  showspaces=false,
  showstringspaces=false,
  showtabs=false,
  stepnumber=1,
  tabsize=2,
  title=\lstname
}

\lstdefinestyle{c_bash}{
 backgroundcolor=\color{light-gray},
 language=bash,
 commentstyle=\color{black},
 columns=flexible
}

\newtheorem{thm}{Theorem}

\begin{document}

\title{Protocol Proxy: \\An FTE-based Covert Channel}

\begin{abstract}
In a hostile network environment, users must communicate without being detected. This involves blending in with the existing traffic. In some cases, a higher degree of secrecy is required. We present a proof-of-concept format transforming encryption (FTE)-based covert channel for tunneling TCP traffic through \textit{protected static} protocols. Protected static protocols are UDP-based protocols with variable fields that cannot be blocked without collateral damage, such as power grid failures. We (1) convert TCP traffic to UDP traffic, (2) introduce observation-based FTE, and (3) model interpacket timing with a deterministic Hidden Markov Model (HMM).  The resulting Protocol Proxy has a very low probability of detection and is an alternative to current covert channels.  We tunnel a TCP session through a UDP protocol and guarantee delivery. Observation-based FTE ensures traffic cannot be detected by traditional rule-based analysis or DPI. A deterministic HMM ensures the Protocol Proxy accurately models interpacket timing to avoid detection by side-channel analysis. Finally, the choice of a \textit{protected static} protocol foils stateful protocol analysis and causes collateral damage with false positives. 
\end{abstract}


\begin{frontmatter}

\author{Jonathan~Oakley\corref{cor1}}
\ead{joakley@g.clemson.edu}

\author{Lu~Yu}

\author{Xingsi~Zhong}

\author{Ganesh~Kumar~Venayagamoorthy}

\author{Richard~Brooks}

\cortext[cor1]{Corresponding Author}
\address{Department of Electrical and Computer Engineering, Clemson University, Clemson SC, USA}
\end{frontmatter}

\section{Introduction}
Traffic analysis classifies network traffic using observable information. Network engineers use traffic analysis to ensure quality of service and identify threats. As a result, the development of hardware and software tools that quickly and effectively classify traffic has been encouraged.  Commercial traffic analysis tools are used by governments to block access to websites that counter their current narrative \cite{heydari2017scalable}. Tools for countering traffic analysis have been developed for both criminal use and covert channels.

Tor is a popular overlay network that routes traffic through three randomly chosen nodes on the Internet. Tor uses nested encryption to ensure messages cannot be intercepted. The client encrypts packets. Each relay node decrypts the outermost layer, revealing another encrypted layer for the next hop to decrypt. By wrapping encryption (like the layers of an onion), it is possible to encrypt traffic so  each node only knows its neighbors. This is not a silver bullet. In contested network environments, it is easy to detect and block Tor \cite{Dingledine2011}. To prevent blocking, pluggable transports (PTs) were developed to obfuscate Tor's traffic patterns.

PT developers must ensure their tools are able to penetrate nation state firewalls while authoritarian governments must determine the optimal defense \cite{garnaev2016security}. Some popular PTs simply wrap encrypted traffic with a new header to allow TLS traffic to pass through firewalls \cite{obfs4,wiley2011dust}. At first, this may seem like an elegant solution, but it is simple to add another firewall rule to block this traffic. This is security through obscurity.

Format Transformation Encryption (FTE) is a form of steganography that translates network traffic into a host protocol\footnote{The host protocol refers to the protocol being mimicked}. Previous FTE implementations used regular expressions \cite{dyer2013protocol} and context free grammars \cite{Dyer2015}. Padding and rerouting has obfuscated traffic and removed side-channels \cite{guan2001netcamo}. In previous work, we used FTE and hidden Markov models (HMMs) to translate traffic flows into DNS requests and responses \cite{fu2016covert,fu2017stealthy} and smart grid sensor traffic \cite{zhong2015stealthy}. Fridrich determined an upper bound on the amount of information that could be steganographically encoded in JPEG images before distortions were visually detected \cite{fridrich2006minimizing}. HMMs have some notable advantages:

\begin{enumerate}
\item data windowing of HMMs \cite{schwier2011methods} makes them effective for both protocol detection and mimicry,
\item tools for differentiating HMMs are well defined \cite{schwier2011methods}, and
\item a normalized metric space can directly measure the quality of protocol mimicry \cite{lu2013normalized}.
\end{enumerate}

We propose security through collateral damage.  Certain protocols are more expensive to block than others. Usually, blocking the wrong TLS stream has little collateral damage other than disgruntled users. 

We chose Synchrophasor traffic for our FTE implementation, but another example is Network Time Protocol (NTP) traffic.  We use Synchrophasor traffic because we have access to Clemson's Real-Time Power and Intelligent System (RTPIS) Laboratory \cite{powerlab} and real Synchrophasor traffic.\footnote{Synchrophasor traffic is a UDP-based protocol generated by Phasor Measurement Units (PMUs) that contains alternating current phase measurements. This protocol is used to balance the load at different points in the power grid.}  To accomplish this transformation, we made the following novel contributions:

\begin{enumerate}
\item An architecture to tunnel TCP traffic through UDP traffic.
\item Real-time observation-based format transforming encryption (FTE) \cite{zhong2015stealthy}.
\item A theoretical upper bound on the channel capacity of observation-based FTE.
\item Emulating the packet timing of a host protocol.
\item A proxy capable of tunneling SSH (TCP) through power-grid Synchrophasor traffic (UDP) in a statistically indistinguishable manner.
\end{enumerate}

In Section~\ref{sec:related}, we introduce related work. In Section~\ref{sec:undetectable}, we justify observation-based FTE as a undetectable communication channel. In Section~\ref{sec:hmm}, we provide the mathematical background behind deterministic HMMs. In Section~\ref{sec:arch}, we provide the system architecture and justify our design decisions. Figure~\ref{fig:abstract} shows our high-level system architecture: Section~\ref{sec:arch-fte} describes observation-based format transforming encryption using our novel method. Section~\ref{sec:arch-timing} describes massaging packet timing with a deterministic HMM inferred using the method described in Section~\ref{sec:related-inferring}, and Section~\ref{sec:arch-converter} describes how everything fits the overall Protocol Proxy architecture. Section~\ref{sec:exp} provides the experimental setup, and Section~\ref{sec:results} details our results. Finally, we provide our closing thoughts and future work in Section~\ref{sec:conclusion}.

\begin{figure*}[!t]
\begin{center}
\includegraphics[scale=1]{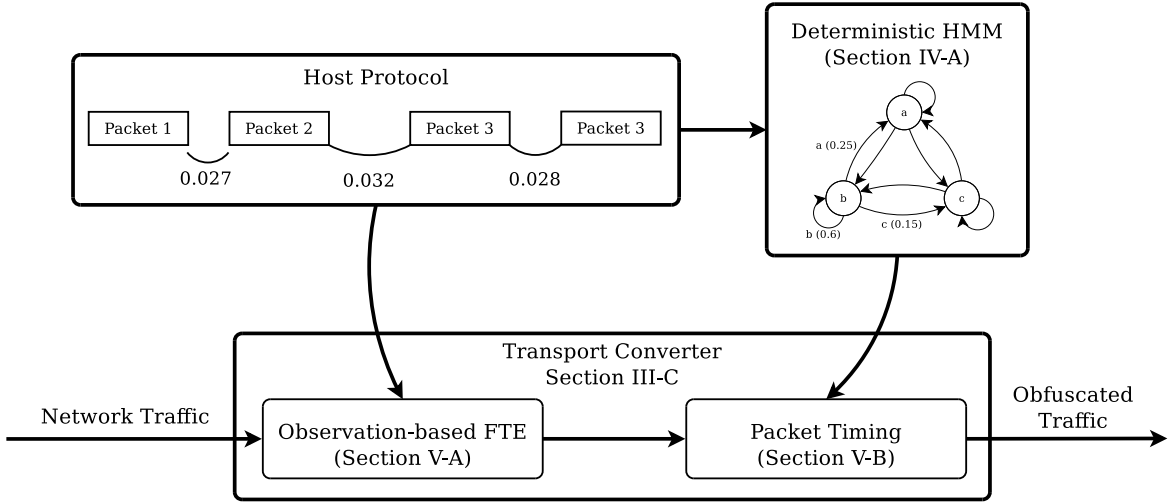}
\caption{Protocol Proxy architecture with relevant sections indicated.} \label{fig:abstract}
\end{center}
\end{figure*}

\section{Related Work} \label{sec:related}
Creating covert online communication tools has been the focus of many privacy advocacy groups. Since the data in covert channels is encrypted, the goal is to balance probability of detection with throughput based on the desired application \cite{smith2010predictable}.  Timing side-channels are used when low probability of detection is prioritized over throughput \cite{kiyavash2013timing,yao2009study}. 

Named Data Networking (NDN) posits an alternative internet architecture based on content delivery \cite{zhang2014named}. Consumers express \textit{interest} in a topic to NDN routers. These NDN routers check their \textit{Content Stores} to see if an interest is cached. If the information is not available, the router adds the interest to a \textit{Pending Interest Table} and forwards the interest upstream according to the \textit{Forwarding Information Base} and \textit{Forwarding Strategy}. Tsudik et. al proposed an anonymous communication network using NDN \cite{tsudik2016ac3n}. Cui et al. proposed a model for preventing censorship using \textit{smart} NDN routers \cite{cui2016defend}. Neither of these solutions addresses the issue of covert communications in a contested environment.

Ambrosin et. al proposed a method for delay-based covert communication using cache techniques \cite{ambrosin2014covert}. Given a sender and receiver share an NDN router at some point in the network, the sender and receiver can communicate using the round trip time (RTT) of the receiver's interest requests. The sender and receiver agree on $C_0$ and $C_1$ out-of-band. The sender requests a certain interest, $C_b$, and the receiver receives the message by requesting both $C_0$ and $C_1$. By comparing the RTT of both $C_0$ and $C_1$, it is possible to determine which interest the sender requested--$C_b$ will have a shorter RTT since it already exists in the router's cache. While this covert channel is interesting, it would be easy to detect in Iran or China since CCNx (the NDN implementation) is not widely used. The traffic would be anomalous in that environment and could be used to identify users before being blocked  \cite{mosko2014ccnx}.

Tor's \cite{tor} anonymity network wraps network traffic in layers of encryption. Each layer can only be decrypted by the next hop in the \textit{onion} network. While it provides anonymous access to the Internet, the Tor protocol is easy to detect and block \cite{winter2012great,Dingledine2011}. Undetectable communication was not one of Tor's goals, but it spawned the Pluggable Transport project to address this challenge and encourage the development of other covert communication tools \cite{ptv2}.

Pluggable Transports (PTs) \cite{pluggable-transports} address this concern. PTs offer a generic way to obfuscate traffic. \textit{Shape-shifting} PTs transform traffic into a different protocol. SkypeMorph \cite{Mohajer2012} makes network traffic resemble a Skype session. StegoTorus \cite{Weinberg2012} demultiplexes connections to avoid traffic analysis and uses steganography to hide information in different protocols (including Skype). In \textit{The parrot is dead: Observing unobservable network communications},  Houmansadr et al. \cite{houmansadr2013parrot} found  both approaches fell short of true protocol mimicry. In both cases, handshake packets were incorrect. Other flaws were noted with StegoTorus's implementation of HTTP steganography \cite{houmansadr2013parrot}. Censorspoofer mimics the Ekiga VoIP software, but it also falls short of mimicking protocol intricacies \cite{houmansadr2013parrot}.

A number of PTs \textit{scramble} traffic to remove fingerprints. Obfs2 \cite{obfsproxy}, Obfs3 \cite{obfsproxy}, Obfs4 \cite{obfs4}, and ScrambleSuite \cite{winter2013scramblesuit} each attempt to remove a network fingerprint by \textit{scrambling} the data. Dust2 and its previous version (Dust) change statistical properties of traffic to bypass firewalls \cite{wiley2011dust}. With technologies like software defined networking (SDN), these statistical PTs will likely be blocked by adaptive firewalls. 

Recent PTs use \textit{domain fronting}. Traffic is sent to a benign destination (Google, Amazon, Azure, etc.) and allowed through the firewall because blocking such a large domain would cause unintended collateral damage. FlashProxy \cite{moshchuk2008flashproxy}, SnowFlake\cite{snowflake}, and meek \cite{fifield2015blocking} all use variations of \textit{domain fronting}.  This approach has been successful but is not condoned by companies whose domains are being used since it exposes them to potential backlash. 

FTE PTs are a subset of \textit{shape-shifting} PTs that steganographically encode traffic using values typical of the host protocol \cite{dyer2013protocol}. It is best to use a widely adopted protocol, such as DHCP \cite{rios2013covert} or VoIP \cite{schmidt2018exploiting}. Marionette \cite{Dyer2015} is a \textit{shape-shifting} PT that uses a probabilistic context-free grammar (PCFG) and production rules to mimic the host protocol. The PCFG ensures traffic is syntactically and semantically correct and production rules occur at the expected frequency. Determining the appropriate PCFG to model a protocol is an open research question \cite{dyer2013protocol, Dyer2015}.  Marionette ensures interpacket timing, packet size, and session count mimic the host protocol.  

Refraction Networking (TapDance) \cite{Wustrow2014} spoofs the destination IP address. \textit{If} the packet is routed through a decoy router, the true destination IP address is substituted for the spoofed address. Recent work has shown  it may be inexpensive to censor decoy routers \cite{schuchard2012routing}. Alternatively, TARN \cite{yu2017tarn} provides an approach that mixes traffic from different autonomous systems at the software defined exchange (SDX) level. This provides a high level of anonymity and is resistant to a malicious ISP or BGP injection, but it is not realistic for a covert channel. Network-based moving target defense solutions have also been proposed \cite{heydari2017scalable} for covert channels.

GNUnet \cite{gnunet}, I2P \cite{i2p}, and Freenet \cite{freenet} all seek to provide anonymous access to the Internet. GNUnet is a toolbox for developing secure decentralized applications, but widespread censorship is possible \cite{kugler2003analysis}. I2P uses garlic routing (an onion-based routing protocol) to route traffic securely, but I2P is meant to be a self-contained network. I2P can be blocked if an adversary controls a small number of routers in the network and uses traditional IP-based filtering \cite{hoang2018empirical}. I2P routers can be identified because hiding I2P traffic was not a design goal \cite{i2p-threat}. Freenet focuses on using prior knowledge to form connections, and it is possible to passively (or actively) scan the network for nodes \cite{roos2014measuring}. It arguably provides more anonymity, but it is resource intensive. Many governments block access to these tools, which makes the first hop important.

Traditional Virtual Private Networks (VPNs)  are not usually effective in a contested environment because encrypted data can indicate malicious activity \cite{iran}. As a result, Psiphon \cite {psiphon}, Lantern \cite{lantern}, and Ultrasurf \cite{ultrasurf} have started using PTs. With Lantern, traffic is only forwarded through the PT if it is likely to be blocked.

Image steganography is also an effective means of covert communication. Fridrich investigated the relationship between distortion and information capacity \cite{fridrich2006minimizing}. Unfortunately, the model derived in \cite{fridrich2006minimizing} does not directly apply to FTE-based covert channels.

\subsection{Previous Work}
This article extends the work by Zhong et al. in \cite{zhong2015stealthy} where network traffic was manipulated offline as a proof-of-concept. The Protocol Proxy architecture presented in this work is a novel contribution designed to address the issues that result from real-time traffic manipulation. The observation-based FTE algorithm was enhanced from \cite{zhong2015stealthy} to increase throughput. In this work, we present a theoretical bound on the information that can be encoded using observation-based FTE. Zhong et al. \cite{zhong2015stealthy} manipulated packet timing in an offline proof-of-concept by setting the packet timestamp in the PCAP file. In this work, we manipulate packet timing in real-time, which creates additional challenges that are addressed by the novel Protocol Proxy architecture.  Zhong et al. \cite{zhong2015stealthy} did not consider guaranteed delivery for a UDP protocol, which must be addressed when manipulating TCP traffic in realtime. Finally, while Zhong et al. \cite{zhong2015stealthy} used HMMs to manipulate packet timing, they did not consider the need for formal model verification.

\section{Undetectability} \label{sec:undetectable}

Detecting protocol mimicry can be done in several ways: rule-based analysis, deep packet inspection (DPI), stateful protocol analysis, side-channel analysis, and statistical analysis.  If a packet matches a set of predetermined rules, then it is accepted (or rejected). Rules typically look only at the packet headers. They concentrate on IP source, IP destination, source port, destination port, protocol type (TCP or UDP), and several other fields. These fields are available in the packet header without the need for DPI, so this is the first line of defense for high throughput use cases. If packets are destined for IP addresses that belong to a malicious web site, they will be dropped before they leave their respective autonomous systems. Source and destination ports are also used to classify traffic \cite{kim2008internet}.  

DPI classifies traffic by analyzing the protocol payload. Recent developments allow primitive classification of HTTPS (encrypted) traffic \cite{miller2014know}. It has been shown  FTE can bypass both rule-based analysis and DPI \cite{dyer2013protocol, Dyer2015}.

``Stateful'' firewalls avoid certain attacks by only permitting traffic if the traffic obeys the underlying protocol. For instance, TCP traffic is required to complete the handshake before data packets are allowed through the firewall. Housmansadr et al. \cite{houmansadr2013want} used an extension of this idea to identify protocol mimicry. By classifying the states of a host protocol, it is possible to identify poor imitations by identifying discrepancies between observed and expected states. Housmansadr et. al used this to find where SkypeMorph differed from Skype. Once differences were identified, simple rules can identify SkypeMorph traffic. 

Hidden Markov Models (HMMs) have been used in side-channel analysis to identify Synchrophasor traffic in an encrypted VPN \cite{zhong2015side}. HMM inference occurs offline, and it requires a large sample of traffic to build the timing model of the unidentified protocol. In the future, HMM detection will likely be performed online. 

Statistical analysis is a hybrid approach that attempts to use statistical properties to identify traffic. Chaos theory is one approach to statistical detection \cite{zhao2012detecting}. Entropy has also been used with distributed denial of service (DDoS) detection, but entropy can easily be spoofed \cite{ozccelik2015deceiving}. We use statistical analysis to compare the timing models of the original traffic with the traffic we generated.

Our variation of protocol mimicry uses a \textit{protected static} protocol. We use this term to refer to a specific subset of protocols that are prime candidates for protocol mimicry. Static protocols are UDP-based and lack application-layer handshakes (like those in Skype), making them immune to stateful analysis. The final layer of security is choosing protocols that are \textit{protected}. These protocols have high collateral damage for false positives. If Synchrophasor packets are dropped, it can have adverse consequences for the power grid.

\section{Hidden Markov Models} \label{sec:hmm}
A Markov model is a tuple $G = (S, T, P )$ where $S$ is a set of states of a model, $T$ is a set of directed transitions between the states, and $P=\left\{ p(s_i,s_j)\right\} $ is a probability matrix associated with transitions from state $s_i$ to $s_j$ such that:
\begin{equation}
\sum\nolimits_{s_j \in S} p(s_i,s_j) = 1, \forall s_i \in S
\end{equation}

A Markov model satisfies the Markov property, where the next state only depends on the current state. An HMM is a Markov model with unobservable states. A standard HMM \cite{eddy1996hidden, rabiner1989tutorial} has two sets of random processes: one for state transition and the other for symbol outputs. HMMs have been used to effectively model time series data \cite{asadi2016creating}. A deterministic HMM \cite{lu2013normalized, lu2012network, schwier2009pattern} is used in this paper, and it has one random processes for state transitions. Different output symbols are associated with transitions with different probability. This representations is equivalent to the standard HMM \cite{vanluyten2008equivalence,lu2013normalized}.

\subsection{Inferring Deterministic HMMs} \label{sec:related-inferring}
Deterministic HMM inference is depicted in Figure~\ref{fig:hmm-process}. A stream of network packets, Figure~\ref{fig:hmm-process-a}, is observed. The interpacket delay (time between each packet) is calculated, and the values are plotted in a histogram. This histogram is grouped into different states, and these states manifest themselves as peaks in the histogram. In Figure~\ref{fig:hmm-process-b}, there are three peaks. Each peak is given a unique label. The stream of interpacket delays is re-interpreted using the assigned labels. A stream of labels, as shown in Figure~\ref{fig:hmm-process-c}, is used to infer the deterministic HMM shown in Figure~\ref{fig:hmm-process-d}. Each state in the HMM corresponds to a label. The probability of an `a' output expression in state `b' is given by the number of occurrences of the string `ba' divided by the number of occurrences of the string `b'. If there were 1000 occurrences of the string `b', and we know the string `ba' occurred 250 times, then 25\% of the time we transitioned to state `a'. The full process for inferring deterministic HMMs is provided in \cite{griffin2011hybrid, schwier2009zero}. Given a deterministic HMM, it is possible to generate a stream of packet timings.

\begin{figure}
\centering
\subfloat[Stream of incoming packets with timing denoted.]{\includegraphics[width=0.45\textwidth]{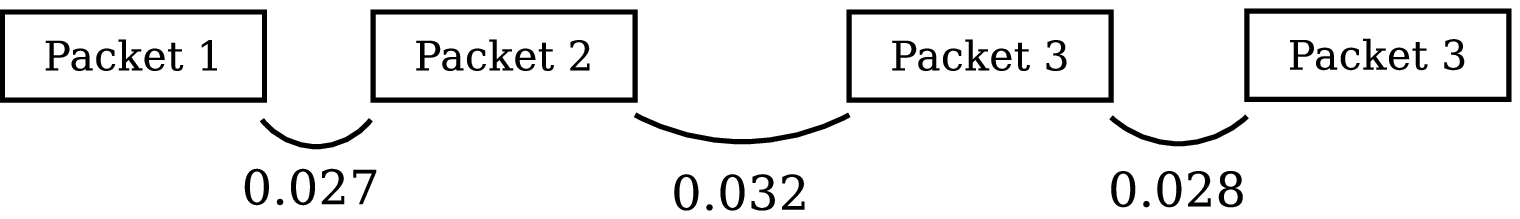}\label{fig:hmm-process-a}}\\
\subfloat[The timing values plotted in a histogram.]{\includegraphics[scale=.3]{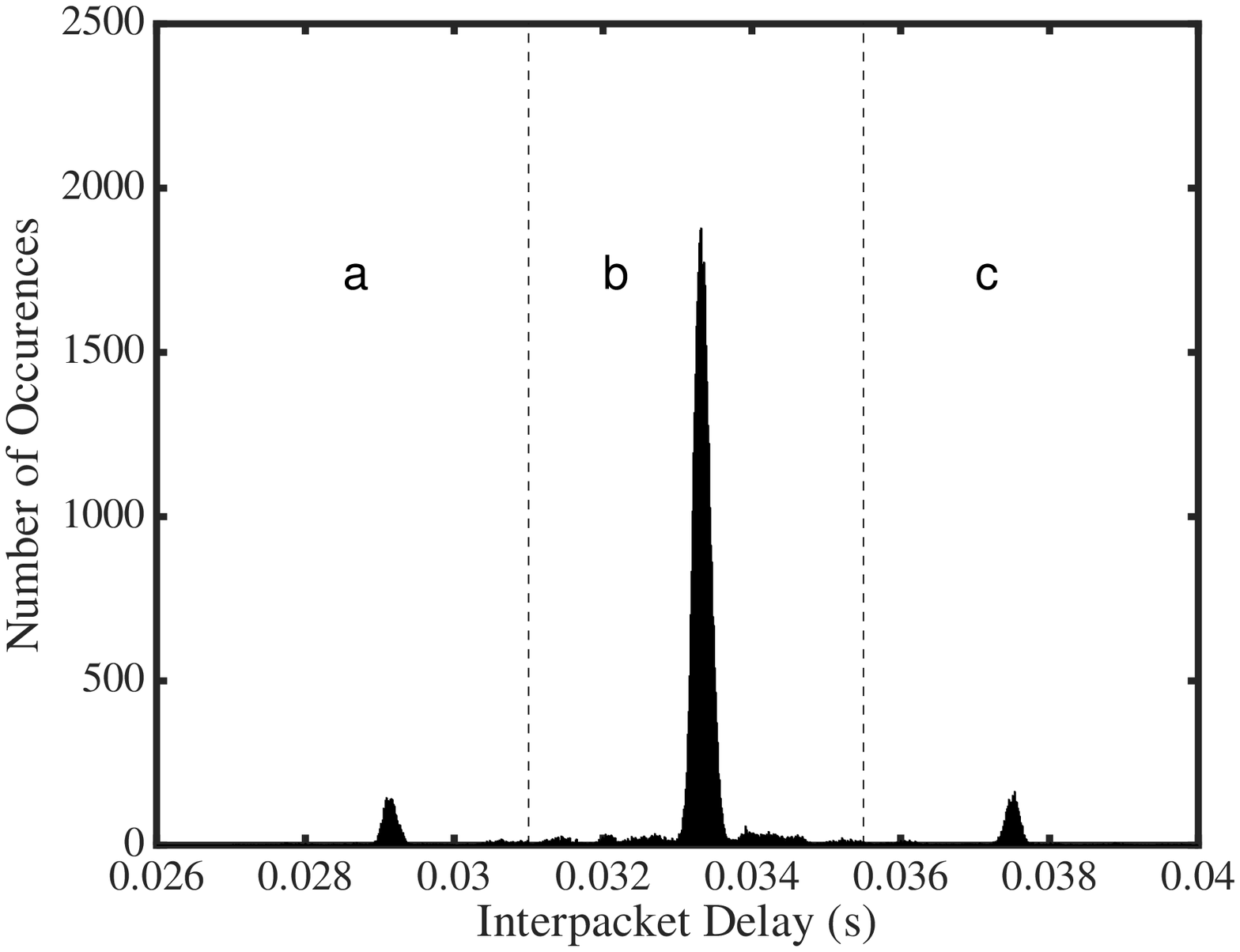} \label{fig:hmm-process-b}}\\
\subfloat[The packet timings converted to a stream labels.]{\includegraphics[scale=.3]{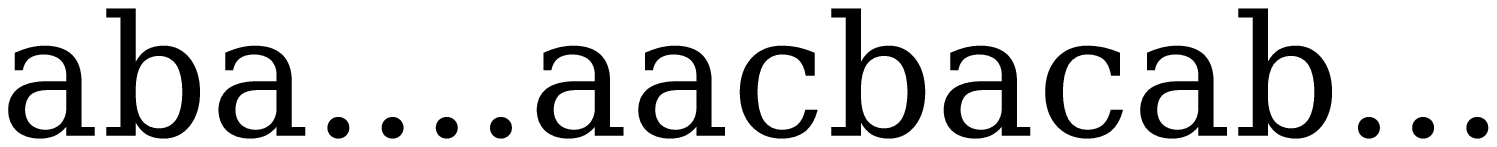}\label{fig:hmm-process-c}}\\
\subfloat[The deterministic HMM inferred from the stream of labels.]{\includegraphics[scale=.3]{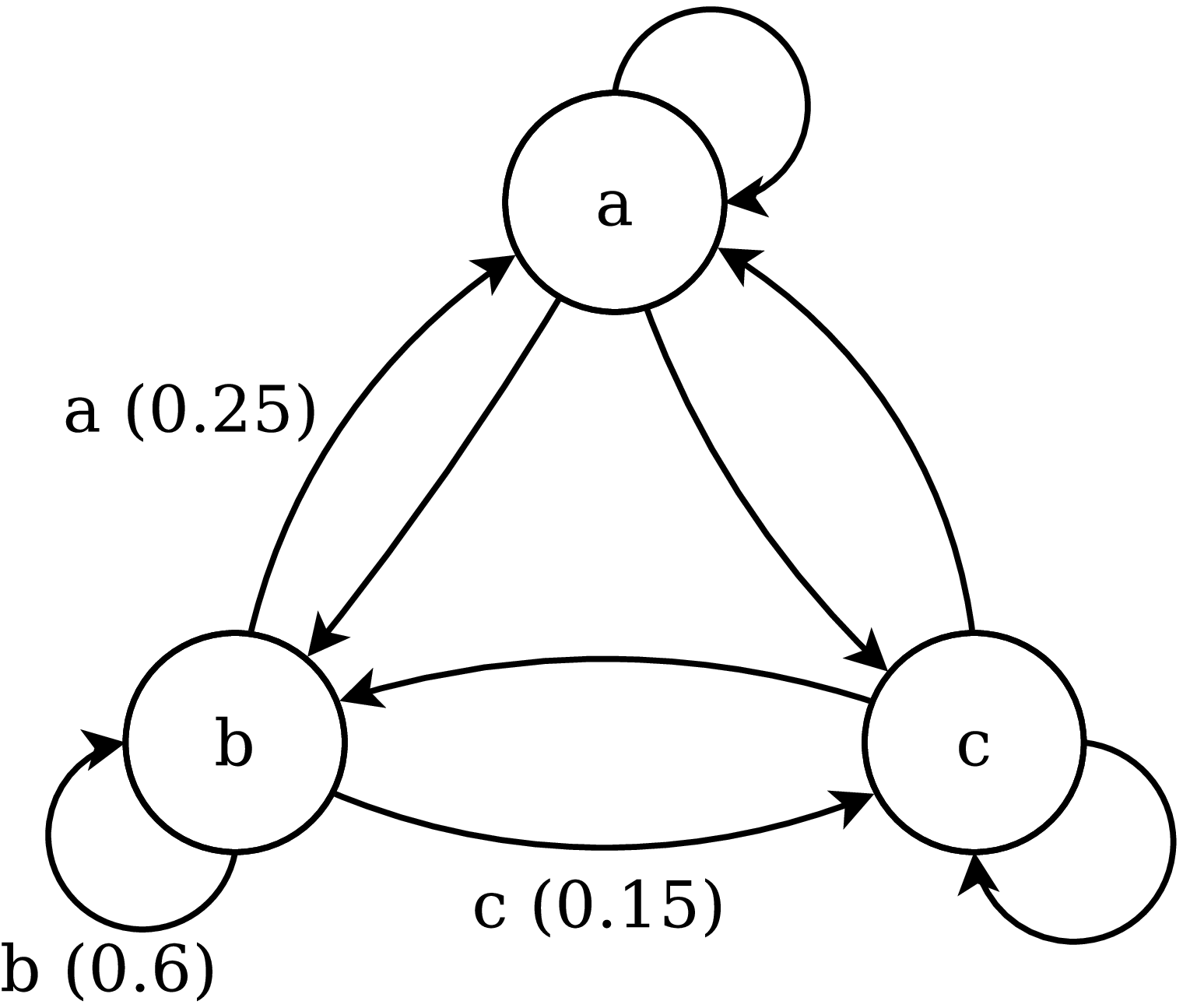}\label{fig:hmm-process-d}}
\caption{Detailed example of how a deterministic HMM is inferred from packet timing.} 
\label{fig:hmm-process}
\end{figure}

\subsection{Comparing Deterministic HMMs}
In \cite{lu2013normalized}, the authors develop a normalized metric space for comparing HMMs, and in \cite{yu2013inferring} the authors show a method for ensuring an inferred HMM is significant. We use an alternative approach that is tailored to this challenge. Before determining whether two deterministic HMMs are equal, it is desirable to ensure the probability distribution functions (PDFs) used to generate the HMM are equal. To do this, we use the two-sample Kolmogornov-Smirnov (KS) test \cite{kstest2}, which tests the null hypothesis (two sets of samples come from the same underlying distribution) against the alternate hypothesis (two sets of samples come from different underlying distributions). The KS statistic is the empirical distribution function $F_n$, defined below.

\begin{equation}\label{eq:ks}
F_n(x) = \frac{1}{n}\sum_{i=1}^{n}I_{(-\infty,x]}(X_i)
\end{equation}

Here, $n$ refers to the number of identically independently distributed samples ($X_i$) taken from the sample space ($X$). Samples ($X_i$) are randomly chosen observations from Figure~\ref{fig:hmm-process-a}. The indicator function, $I_{(-\infty,x]}(X_i)$, is defined in Equation~\eqref{eq:indicator}.

\begin{equation}\label{eq:indicator}
I_{[-\infty,x]}(X_i) = \left\{\begin{split}
&1, \quad &X_i < x \\
&0, \quad &\textrm{otherwise} \\
\end{split}\right.
\end{equation}

The two-sample KS test compares the distance between the two empirical distribution functions using Equation~\eqref{eq:ks2}.

\begin{equation}\label{eq:ks2}
D_{n,m} = \sup_x | F_{1,n}(x) - F_{2,m}(x)|
\end{equation}

The null hypothesis is rejected at the $95\%$ confidence level if the following criterion is met.

\begin{equation}\label{eq:ks-reject}
D_{n,m} > 1.36 \sqrt{\frac{n+m}{nm}}
\end{equation} 

To determine whether two deterministic HMMs are equivalent, it is sufficient to show  all corresponding states in the deterministic HMM are equivalent. If all states are equivalent, the HMMs are equivalent. To show  two states of a deterministic HMM are equivalent,  we use the $\chi^2$ test for homogeneity to test if the probability distributions for outgoing state transitions are statistically equivalent. The generic expression for the $\chi^2$ statistic for homogeneity given $P$ populations and $C$ levels of the categorical variable is shown below.

\begin{equation}
\chi^2 = \sum_{i\in P}\sum_{j\in C} \frac{\left(O_{i,j} - E_{i,j}\right)^2}{E_{i,j}}
\end{equation}

In this representation, $O_{i,j}$ is the number of occurrences observed in the state corresponding to $i$ and the output expression corresponding to $j$. Similarly, $E_{i,j}$ is the number of \textit{expected} occurrences for the combination of state and output expression. The expected number of occurrences is calculated as shown below in Equation~\eqref{eq:expected}.

\begin{equation}\label{eq:expected}
E_{i,j} = \frac{n_in_j}{n}
\end{equation}

Here, $n_i$ is the number of observations in state $i$, $n_j$ is the number of observations at that level of the categorical variable, and $n$ is the sample size.  For threshold testing, the degrees of freedom ($DF$) is given as follows.

\begin{equation}
DF = (P - 1)(C - 1)
\end{equation}

In this work, we compare two states (populations), so $P$ is 2. Therefore, the $DF$ for any given state is simply the number of output expressions ($C$) minus one.

\section{Architecture} \label{sec:arch}
Converting TCP-based Tor traffic to the UDP Synchrophasor protocol requires a number of building blocks that were not present in \cite{zhong2015stealthy}. The TCP packet must be converted to several Synchrophasor packets. The packet timing of outgoing packets must be adjusted to model the timing of Synchrophasor traffic. A transport layer converter is required to tunnel TCP-based protocols through UDP protocols while still maintaining TCP's guaranteed delivery. Finally, these building blocks can be linked together to provide a proof-of-concept that can convert Tor's TCP traffic to UDP Synchrophasor traffic.

\subsection{Observation-based FTE} \label{sec:arch-fte}
Simply sending UDP packets to a specific port isn't enough. Capturing the packet in an analysis tool like Wireshark \cite{wireshark} will reveal the packet is malformed. While this rises to the level of existing obfuscation PTs, it does not solve the problem. Traditional FTE takes the syntax of a protocol and creates a PCFG to map raw binary data to that protocol's syntax \cite{Dyer2012}. Determining the appropriate PCFG to model a protocol is left as an open research question, which makes it unrealistic to deploy \cite{dyer2013protocol, Dyer2015}. 

We propose observation-based FTE, as an alternative. We collect a substantial amount of traffic and record the unique observations for each field in the protocol. If Alice and Bob want to encode information using observation-based FTE with the protocol shown in Figure~\ref{fig:ofte}, they construct a lookup table with each field and an ordered list of observations. This lookup table is a shared out-of-band. High-entropy (encrypted) information can be encoded using observation-based FTE by construction a packet using observations from the host protocol. With the protocol in Figure~\ref{fig:ofte}, three bits can be encoded in the first field. When Alice receives the message and sees `observation 5' in `field 1', she uses the shared lookup table to determine the first three encoded bits are `101' (the binary value of the observation's index). The high-entropy input also ensures each field (and each packet) is independent of the other fields (and packets)$^3$.

\begin{figure}
\begin{center}
\includegraphics[scale=.5]{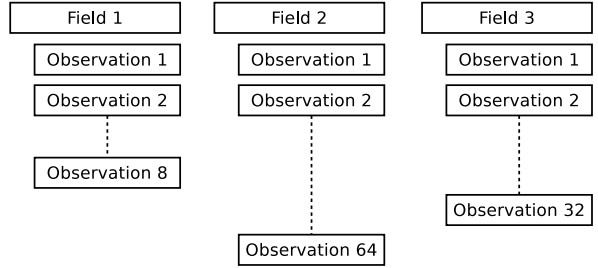}
\caption{Example protocol to illustrate observation-based FTE.} \label{fig:ofte}
\end{center}
\end{figure}

  Zhong et al.'s work \cite{zhong2015stealthy} used a primitive version of observation-based FTE that did not consider the upper bound on the information capacity of an FTE channel. 

\begin{thm}  For a given protocol, the maximum amount of information that can be encoded in a packet using observation-based FTE is given by:
\begin{equation} \label{eq:total_size}
\textrm{S} = \sum_{ \gamma_i \in \Gamma}\textrm{log}_2(|\boldsymbol \gamma_i|)
\end{equation}
Where $\Gamma = \{\gamma_1, \gamma_2, ..., \gamma_n\}$ is the set of $n$ fields in the protocol, and $|\gamma_i|$ is the number of unique observations in that field.
\end{thm}

\begin{proof}
The maximum amount of information that can be encoded in a particular field using observation-based FTE is given by the Shannon entropy of that field. 

\begin{equation} \label{eq:entropy}
H(\gamma_i) = -\sum_{x \in \gamma_i}p(x)\textrm{log}_2(p(x))
\end{equation}

Each stream of $n$ bits is equally likely\footnote{Since the data being mapped to the protocol is encrypted using AES encryption and AES produces a high-entropy bitstream \cite{lyda2007using}, we can assume 0 and 1 are equally likely in practice.}. Therefore, the choice of each observation is equally likely. This simplifies Equation~\eqref{eq:entropy} as follows.

\begin{equation} \label{eq:entropy}
\begin{split}
H(\gamma_i) &= -\sum_{x \in \gamma_i}\frac{1}{|\gamma_i|}\textrm{log}_2\left(\frac{1}{|\gamma_i|}\right) \\
&= -|\gamma_i|\frac{1}{|\gamma_i|}\textrm{log}_2\left(\frac{1}{|\gamma_i|}\right)\\
&= -\textrm{log}_2\left(\frac{1}{|\gamma_i|}\right)\\
&= \textrm{log}_2(|\gamma_i|)
\end{split}
\end{equation}

The amount of information that can be encoded in a single packet is the sum of the information that can be encoded in each field in the packet.

\begin{equation} \label{eq:total_size}
\textrm{S} = \sum_{\gamma_i \in \Gamma}\textrm{log}_2(|\gamma_i|)
\end{equation}
\end{proof}

Performing these calculation on the Synchrophasor protocol yields 516 bits that can be encoded in a single UDP packet. Since this is smaller than the typical TCP packet, it is necessary to segment TCP packets for transmission. The optimal average goodput (G\textsubscript{avg}) can be calculated with Equation~\eqref{eq:goodput}, where S is found using Equation~\eqref{eq:total_size}, and T\textsubscript{avg} is the average interpacket delay, which is 0.03334 seconds for Synchrophasor traffic. This yields an theoretical average goodput of 15,477 bits per second.

\begin{equation}\label{eq:goodput}
\textrm{G}_\textrm{avg} = \frac{\textrm{S}}{\textrm{T}_\textrm{avg}}
\end{equation}

\begin{figure*}[!t]
\begin{center}
\includegraphics[scale=1]{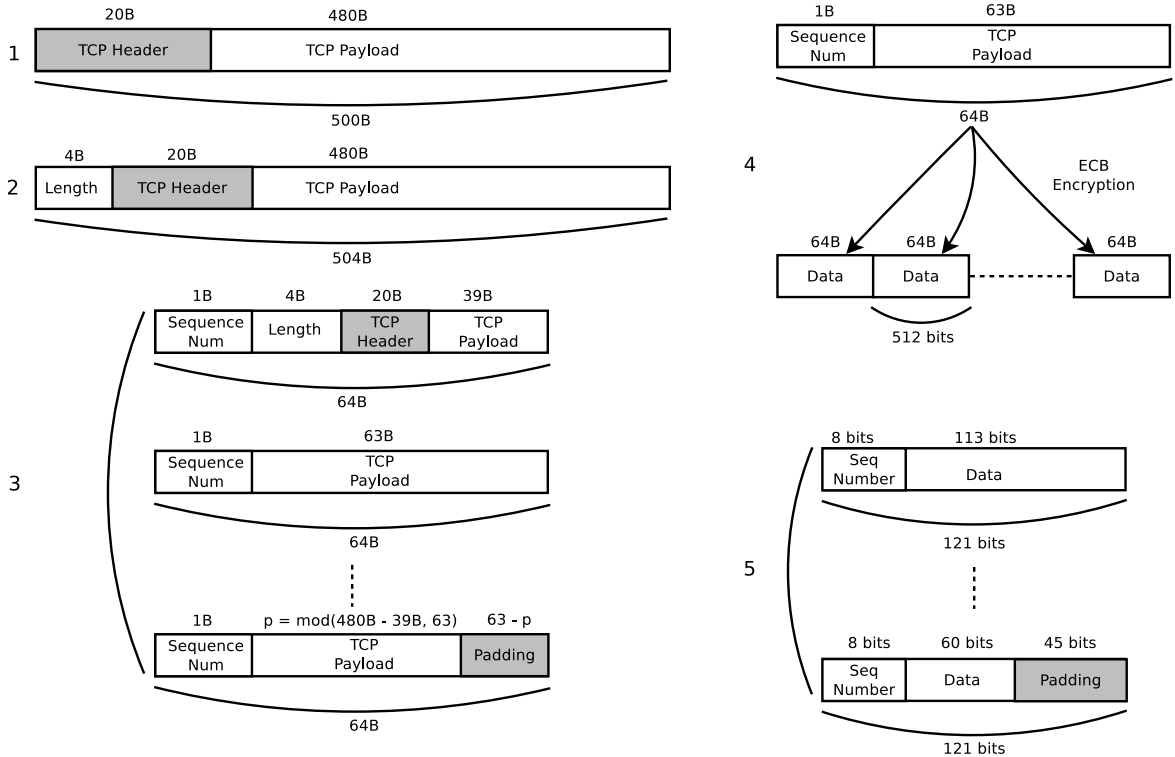}
\caption{Process for segmenting TCP packets for transmission.} \label{fig:segmentation}
\end{center}
\end{figure*}

Segmentation is shown in Figure~\ref{fig:segmentation}, where each item below is indicated in the figure:

\begin{enumerate}
\item The original TCP packet is taken.
\item The packet length is prepended to the beginning of packet as a four byte unsigned integer.
\item The packet is broken into 63 byte chunks, and each chunk is prepended with a one byte sequence numbers for a total of 64 bytes. The sequence number allows the chunks to be reassembled later into the original TCP packet. Depending on the size of the packet, it is possible there will not be enough payload data to fill the final chunk. In this case, random data is appended to the end.
\item Each 64 byte chunk is encrypted with Electronic Code Book (ECB) encryption\footnote{ECB is used because packets may arrive out of order, which makes cipher block chain impractical.}.
\item The observation-based FTE encodes each 64 byte (512 bit) chunk into a 516 bit UDP payloads using the method previously discussed.
\end{enumerate}

\subsection{Packet Timing}\label{sec:arch-timing}
The HMM timing model described in Section~\ref{sec:hmm} was given to the Protocol Proxy, which queries the timing model for a timing value. When the model is queried, it examines its current state, yields an output expression based on the probability distribution of the current state, and chooses a timing value from the output expression group. The model advances to the chosen state. The Protocol Proxy waits for the allotted time before sending a packet. If there are no packets to send, random data is encoded  and sent. These packets are dropped by the server. Sending placeholder packets ensures the correct timing model is emulated when there are no packets to transmit.

This approach to packet timing differs from the NDN covert channel proposed in \cite{ambrosin2014covert}. Our approach does not convey any information using interpacket delays. Interpacket delays can often be an indicator of a covert channel, so we ensure the timing of the covert channel is statistically identical to the host protocol.

\subsection{Protocol Proxy} \label{sec:arch-converter}
The Protocol Proxy combines the previous elements to create a covert channel using the host protocol.

\begin{figure*}[!t]
\begin{center}
\includegraphics[width=1\textwidth]{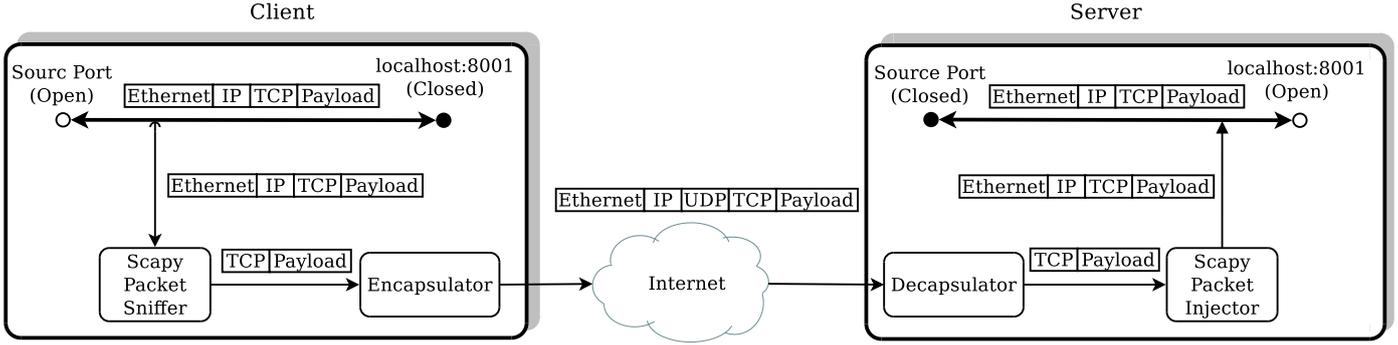}
\caption{Architecture for the Protocol Proxy transport converter.} \label{fig:converter}
\end{center}
\end{figure*}

We use Scapy \cite{scapy}\footnote{Scapy is a Python library for packet capture, manipulation, and injection} to capture TCP packets. These captured packet includes both the TCP header and payload. Applications using the Protocol Proxy are configured to send traffic to a closed loopback port, and Scapy sniffs the loopback interface looking for traffic destined to that port.  By default, when a closed port receives a TCP packet, it responds with a TCP RST (reset) packet. The RST packet immediately terminates the connection and ceases all communication with the other host. This is a low-level response dictated by the host firewall. The transport converter requires this RST packet not be sent, so the internal firewall (\verb!iptables! on Linux) was modified accordingly. The following command disables RST packets for all ports on Linux.

\begin{lstlisting}[style=c_text,caption=The iptables command to disable RST packets.,label=lst:iptables]
sudo iptables -A OUTPUT -p tcp --tcp-flags RST \
	RST -j DROP
\end{lstlisting}

In Figure~\ref{fig:converter}, an application on the client sends a packet to a closed local port (8001 in our example). Scapy sniffs the network stack and captures the entire packet.  The Ethernet and IP layers are stripped, and the packet is transformed into a UDP payload using observation-based FTE, as described in Section~\ref{sec:arch-fte}. After transforming the packet into the host protocol's payload, new Ethernet, IP, and UDP layers are generated so the packet can be sent across the network to the server.  All the UDP packets are sent to the port that corresponds to the emulated protocol. The timing of outgoing packets is determined by the HMM described in Section~\ref{sec:arch-timing}.

On the server, the Protocol Proxy listens on the predetermined port for the UDP packet. Once it receives the UDP packet, it takes the UDP payload, reverses the observation-base FTE transformation described in Section~\ref{sec:arch-fte}, and sends it to a Scapy packet injector. The Scapy packet injector creates new Ethernet and IP layers to make the packet look like it originated from the local machine. Then, Scapy injects the new packet into the local network stack, and it is received by the local application. This provides a bidirectional covert channel that uses actual observations from host protocol.

Since the Protocol Proxy is accessed via an open port, it is trivial to integrate with Tor. Tor natively supports tunneling outgoing connections through a SOCKS proxy, so the Tor client is configured to use the Protocol Proxy port as a SOCKS proxy. On the server, a SOCKS server listens for traffic forwarded from the Protocol Proxy.  

\section{Experiment Setup} \label{sec:exp}

\begin{figure*}
\centering
\includegraphics[scale=0.7]{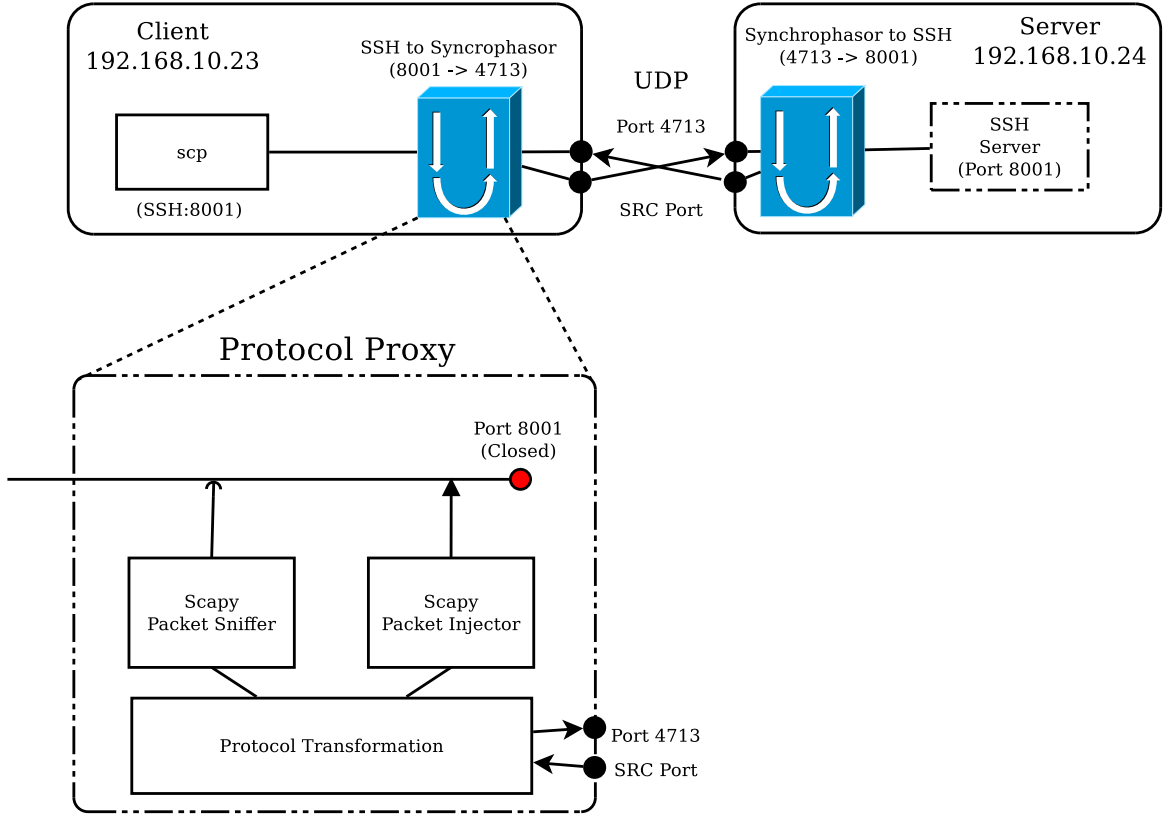}
\caption{The Protocol Proxy integrated for use with SCP.}
\label{fig:complete-system}
\end{figure*}

Over 770,000 samples were collected from the PMUs in Clemson's RTPIS Laboratory \cite{powerlab}. These samples were used to model the timing of the protocol as described in Section~\ref{sec:arch-timing}.  All testing was done in Clemson's security lab with clean installations of Arch Linux (kernel version 4.17.2-1). The experiment setup is shown in Figure~\ref{fig:complete-system}. We used \verb!scp! to transfer data over the Protocol Proxy to an SSH server on a remote machine. The Protocol Proxy server was launched using the following command.  

\begin{lstlisting}[style=c_bash]
server# protocol_proxy server 192.168.10.23 8001
\end{lstlisting}

Since the Protocol Proxy requires access to raw sockets, it must be executed by a privileged user. The `server' option tells the Protocol Proxy to expect packets originating from the specified port (8001). The IP address (192.168.10.23) is the IP address of the client that will connect to the server. The next step is to ensure the kernel does not send a reset packet (TCP RST) when packets are sent to the closed port (8001). This is done by executing the \verb!iptables! command shown in Listing~\ref{lst:iptables}. Next, the SSH server is set to listen to port 8001 for incoming connections in the \verb!/etc/ssh/sshd_config! file.  SSH (OpenSSH\_7.7p1) was used for the server. It was necessary to configure a non-standard port to avoid conflict when forwarding the traffic through the Protocol Proxy. The client was launched with the following command.

\begin{lstlisting}[style=c_bash]
client# protocol_proxy client 192.168.10.24 8001
\end{lstlisting}

Again, the application must be executed by a privileged user. The `client' option tells the transport to expect packets destined for the specified port (8001). The IP address (192.168.10.24) is the IP address of the host executing the program. As with the server, the client does not open the port, so the rules in Listing~\ref{lst:iptables} must be applied to the client to ensure TCP connections are not prematurely terminated.

A one kilobyte data file was transferred from the client to the server using \verb!scp! as shown below.

\begin{lstlisting}[style=c_bash]
client# scp -P 8001 file 127.0.0.1:file
\end{lstlisting}

The destination port was set to 8001 on the client, which was the local port being forwarded to the Protocol Proxy. The traffic between the client to the server was captured, and another HMM was inferred from this generated traffic. This second HMM was compared via the $\chi^2$-test to the original HMM used by the Protocol Proxy. Finally, the baseline goodput was measured by recording the time required to transfer the 1 kilobyte data file.

\section{Results} \label{sec:results}
Figure~\ref{fig:original-50k-hist-labeled} shows the histogram of interpacket delay times for the Synchrophasor traffic captured in Clemson's RTPIS laboratory.  The output expressions are labeled in the histogram according to the prominent peaks. Using the techniques described in Section~\ref{sec:hmm}, the deterministic HMM in Figure~\ref{fig:hmm-original} was inferred.

\begin{figure}
\centering
\includegraphics[scale=0.425]{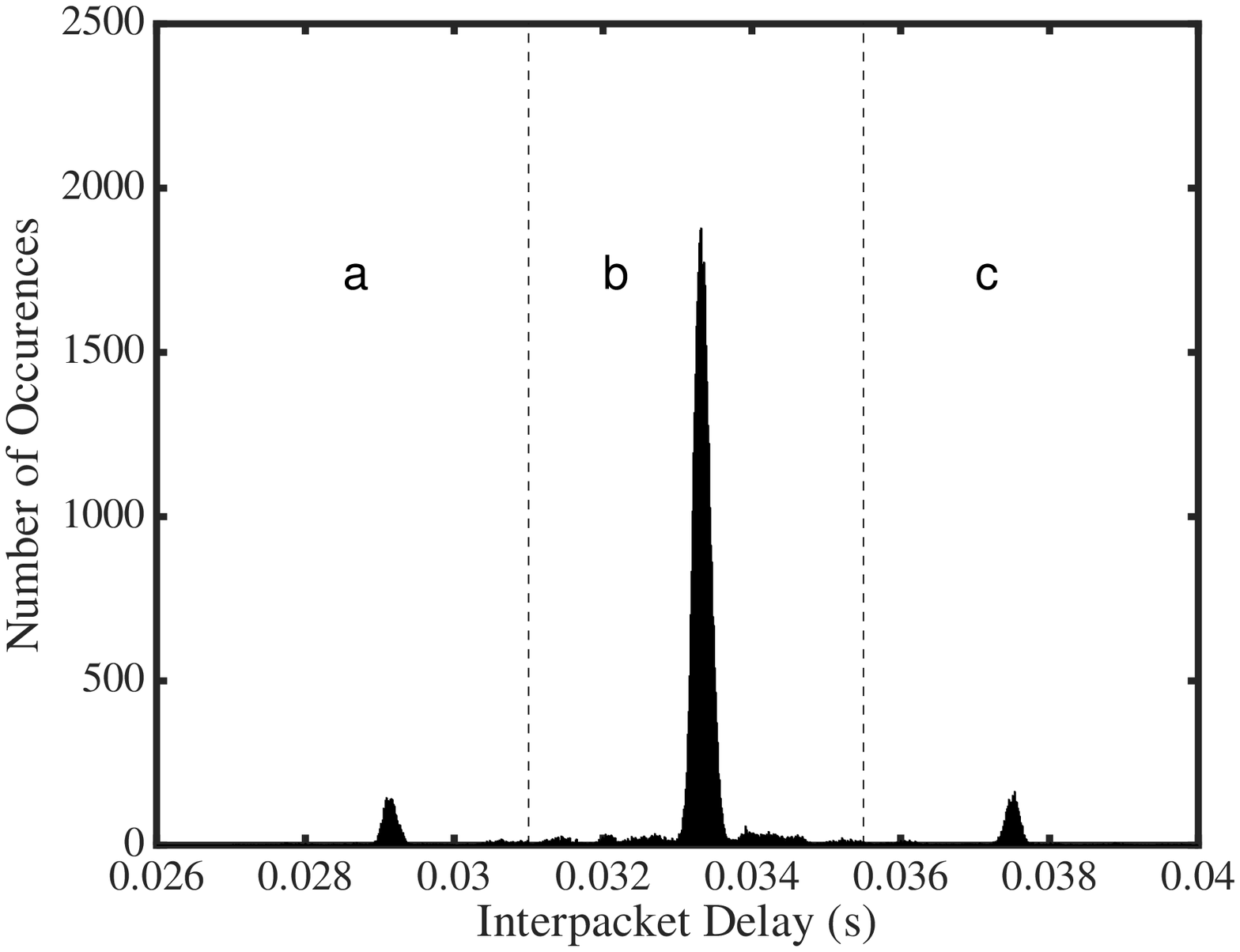}
\caption{Histogram of the interpacket delay of real Synchrophasor traffic with states labeled.}
\label{fig:original-50k-hist-labeled}
\end{figure}

\begin{figure}
\centering
\includegraphics[scale=0.75]{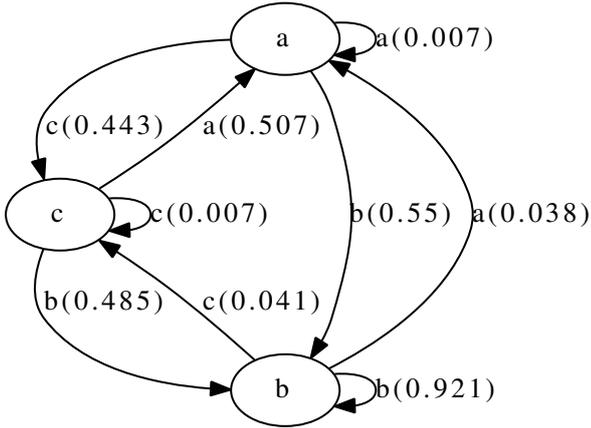}
\caption{HMM generated from the interpacket delay of Synchrophasor traffic.}
\label{fig:hmm-original}
\end{figure}

With this HMM, it was then possible to generate Synchrophasor traffic with observation-based FTE and accurate timing. Figure~\ref{fig:wireshark-traffic} shows a Wireshark deconstruction of the traffic generated with our Protocol Proxy. Wireshark correctly identifies the Protocol Proxy traffic as Synchrophasor traffic and is able to parse the field values from the payload. The checksum is also correctly calculated. This FTE-generated PMU traffic is accepted by the Phasor Data Concentrators--the hardware PMU datastore \cite{zhong2015stealthy}. Since packets generated by the Protocol Proxy only use previously observed field values, the Protocol Proxy is syntactically equivalent to the host protocol. Any rule that would detect Protocol Proxy traffic would have false-positives on legitimate PMU traffic, and these false-positives would lead to significant collateral damage.

\begin{figure}
\centering
\includegraphics[scale=0.5]{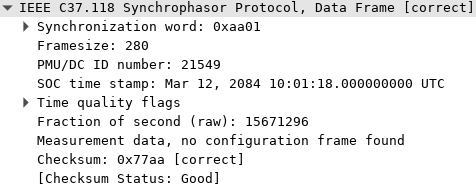}
\caption{Wireshark decoding of the Protocol Proxy traffic.}
\label{fig:wireshark-traffic}
\end{figure}

Figure~\ref{fig:trial6-hist-labeled} shows the histogram of interpacket delay times for the generated traffic with the output expressions labeled.  Figure~\ref{fig:hmm-generated} shows the deterministic HMM inferred from the histogram to model the timing patterns of the Protocol Proxy traffic. Visually, this model appears almost identical to the model used to generate the traffic.

\begin{figure}
\centering
\includegraphics[scale=0.425]{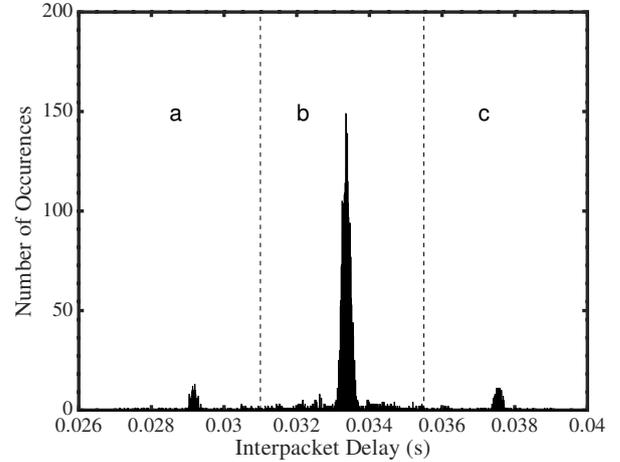}
\caption{Histogram of the interpacket delay of generated Synchrophasor traffic with states labeled.}
\label{fig:trial6-hist-labeled}
\end{figure}

\begin{figure}
\centering
\includegraphics[scale=0.75]{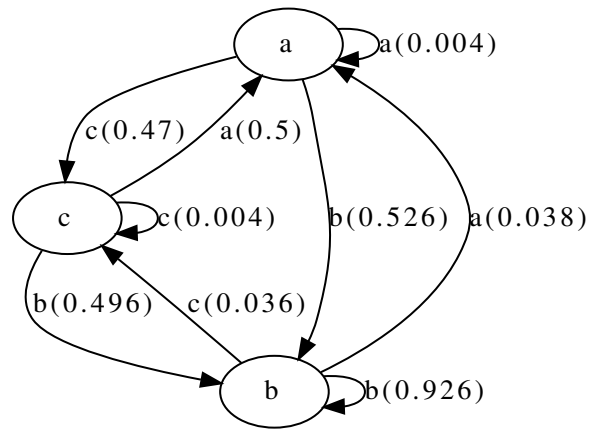}
\caption{HMM generated from the interpacket delay of the generated Synchrophasor traffic.}
\label{fig:hmm-generated}
\end{figure}

Before determining if the two deterministic HMMs were equal, the two-sample KS test was applied to the two distributions (shown in Figure~\ref{fig:original-50k-hist-labeled} and Figure~\ref{fig:trial6-hist-labeled}). To apply this test, we randomly sample each distribution 100 times and apply the test using the two sets of samples. The p-value for the two-sample KS test was found to be 0.21, so with a threshold of 0.05, we fail to reject the null hypothesis. The interpacket delay times of the Protocol Proxy are from the probability distribution as the interpacket delay times of the original Synchrophasor traffic.

To determine if the two deterministic HMMs were equal, the HMMs were checked for state-wise equality using the $\chi^2$ test for homogeneity. The p-values for the $\chi^2$ test are shown in Table~\ref{tab:x2-test}. The first comparison (inferred-inferred) infers two HMMs using 10,000 samples and a random starting point in the original traffic. From these values, we fail to reject the null hypothesis (with an $\alpha$ value of $0.05$) for every state and are left to conclude the traffic is \textit{homogeneous}, which means it does not change over time. The second comparison (generated-inferred) infers one HMM from the Protocol Proxy traffic and another HMM from the original Synchrophasor traffic.  From these values, we fail to reject the null hypothesis (with an $\alpha$ value of $0.05$) for every state and are left to conclude the traffic from the Protocol Proxy is equivalent to the \textit{homogeneous} Synchrophasor traffic.

\begin{table}
\small
\centering
\caption{State-wise $\chi^2$ test for homogeneity comparing HMMs.}
\label{tab:x2-test}
\begin{tabular}{|c|c|c|}
\hline
State Comparison & Inferred-Inferred & Generated-Inferred \\
& (p-value) &(p-value)\\
\hline
\hline
a-a & 0.75 & 0.82 \\
\hline
b-b & 0.19 & 0.37 \\
\hline
c-c & 0.06 & 0.15 \\
\hline
\end{tabular}
\end{table}

The baseline goodput (link speed) was determined to be 54 Mbps, while the goodput through the PMU Protocol Proxy was found to be 182.2 bits per second. These values are compared to the theoretical goodput in Table~\ref{tab:goodput}. The difference between theoretical and observed goodput is attributed to retransmission and packet overhead (the TCP header is sent through the Protocol Proxy). We measured the goodput through Tor to be around 7.31 Mbps using an online speed test\footnote{\url{https://speedof.me}}.

\begin{table}
\centering
\caption{ Comparison of observed and theoretical goodput through the Protocol Proxy.}
\label{tab:goodput}
\begin{tabular}{|c|c|c|c|}
\cline{2-4}
\multicolumn{1}{c|}{}&Baseline & Theoretical & Observed\\
\cline{2-4}
\hline
Goodput & 54 Mbps & 15,477 bps & 182 bps \\
\hline
\end{tabular}
\end{table}

\section{Conclusion and Future Work} \label{sec:conclusion}
Covert communication techniques must evade several types of threats: rule-based detection, DPI, stateful protocol analysis, side-channel analysis, and statistical analysis.  To address these issues, we presented a novel approach for tunneling TCP traffic through \textit{protected static} protocols. Protected protocols have collateral damage associated with false positives, and static protocols are deterministic UDP-based protocols whose pattern never changes. To accomplish this, we introduced (1) an architecture to convert TCP traffic to UDP traffic\footnote{The conversion from TCP to UDP is not intended to provide an additional layer of security. It is necessary because the host protocol is UDP-based, and it is an open challenge in PT development.}, (2) observation-based FTE to mimic the host protocol's payload, and (3) a deterministic HMM to model the host protocol's interpacket timing. The Wireshark packet capture in Figure~\ref{fig:wireshark-traffic} illustrates objectives (1) and (2)--Protocol Proxy traffic is syntactically equivalent to the host protocol. The $\chi^2$ test shows the timing generated by the Protocol Proxy is statistically equivalent to the host protocol.

This extends the simulation present in \cite{zhong2015stealthy} by developing an architecture to perform the protocol transformations in real-time. The Protocol Proxy provides a universal interface for connecting applications. A theoretical upper-bound for the information capacity of an observation-based FTE channel was derived, which facilitated the improvement of the previous observation-based FTE implementation. During testing, it was found some Linux distributions were incapable of emulating the timing of the host protocol. Arch Linux was chosen for its ability to consistently emulate timing. The cause of this discrepancy is currently being investigated. The Protocol Proxy does not transmit information using the interpacket delay as with the NDN-based covert channel \cite{ambrosin2014covert}, and the presence of NDN traffic itself could be enough to indicate a covert channel.

Future work will provide an in-depth security analysis of the Protocol Proxy, decrease the overhead associated with packet retransmissions, and adapt the Protocol Proxy to Tor's PT version 2.0 specification \cite{ptv2}. While the Protocol Proxy's goodput was significantly slower than Tor, the use-case is situations when any anomalous traffic could have negative repercussions. This trade-off between performance and detection is justified in those cases. We will also investigate a number of performance improvements. For instance, the entire TCP packet is currently transmitted through the proxy, but in reality much less data is required for most packets. The Protocol Proxy approach may be vulnerable to semantic analysis using probabilistic context-free grammars, but these techniques are not currently used in the wild for scalability reasons.

The Protocol Proxy  is a viable alternative to traditional transports when heightened anonymity is required.  While there are a number of improvements that will increase throughput, these preliminary results show  it is possible to tunnel a TCP session through a UDP protocol and maintain TCP's guaranteed delivery. Observation-based FTE extends this and ensures the traffic will not be detected by rule-based analysis or DPI. Furthermore, a deterministic HMM ensures the Protocol Proxy accurately models interpacket timing and avoids detection by side-channel analysis. Finally, the choice of a \textit{protected static} protocol ensures  stateful protocol analysis is useless while raising the collateral damage associated with false positives. 

\section*{Acknowledgment}
This material is based upon work supported by, or in part by, the National Science Foundation grants CNS-1049765, OAC-1547245, and CNS-1544910. The U.S. Government is authorized to reproduce and distribute reprints for Governmental purposes notwithstanding any copyright notation thereon. The authors gratefully acknowledge this support and take responsibility for the contents of this report. The views and conclusions contained herein are those of the authors and should not be interpreted as necessarily representing the official policies or endorsements, either expressed or implied, of the National Science Foundation, or the U.S. Government.

\bibliographystyle{IEEEtran}
\bibliography{ProtocolProxy_references}

\end{document}